\documentclass[runningheads,a4paper]{llncs}
\usepackage{epsfig}
\usepackage{algorithmic}
\usepackage{algorithm} 
\usepackage{graphicx}
\usepackage{caption}
\usepackage{amssymb}
\usepackage{geometry}
\usepackage{color}
\usepackage{amsmath}
\usepackage{mathtools}
\usepackage{cite}

\usepackage{versions}
\excludeversion{INUTILE}
\includeversion{BACK}\includeversion{LONG}\includeversion{VLONG}\includeversion{PROOF}\excludeversion{SHORT}\excludeversion{TODO}
\markversion{CARLOS}

\usepackage{microtype}

\begin{document}
\hyphenation{pro-blems sol-ving complexi-ty dia-gram poly-gonal analy-sis ins-tance des-cribed va-lues}

\pagestyle{headings}  
\addtocmark{Synergistic Computation of Planar Maxima and Convex Hull}

\mainmatter 
\title{ Synergistic Computation \\of Planar Maxima and Convex Hull }

\titlerunning{Synergistic Computation of Planar Maxima and Convex Hull}
\author{
  J\'er\'emy Barbay\inst{1} 
  \and 
  Carlos Ochoa\inst{1}\thanks{Corresponding author.}
}

\institute{
  Departamento de Ciencias de la Computaci\'on, Universidad de Chile, Chile\\ 
  \email{jeremy@barbay.cl, cochoa@dcc.uchile.cl.}
}

\maketitle              

\begin{abstract}
    Refinements of the worst case complexity over instances of fixed input size consider the input order or the input structure, but rarely both at the same time. Barbay et al. [2016] described ``synergistic'' solutions on multisets, which take advantage of the input order and the input structure, such as to asymptotically outperform any comparable solution which takes advantage only of one of those features. We consider the extension of their results to the computation of the \textsc{Maxima Set} and the \textsc{Convex Hull} of a set of planar points. After revisiting and improving previous approaches taking advantage only of the input order or of the input structure, we describe synergistic solutions taking optimally advantage of various notions of the input order and input structure in the plane. As intermediate results, we describe and analyze the first adaptive algorithms for \textsc{Merging Maxima} and \textsc{Merging Convex Hulls}.
\end{abstract}

\begin{center}
  \begin{minipage}{.9\textwidth}
  \end{minipage}
\end{center}

\section{Introduction}

One way to close the gap between practical performance and the worst case complexity over instances of fixed input size is to refine the later, considering smaller classes of instances. Such measures of difficulty can be seen along two axis: some depend of the \emph{structure} of the input, such as the repetitions in a multiset~\cite{1976-JComp-SortingAndSearchingInMultisets-MunroSpira}, or the positions of the input points in the plane~\cite{2009-FOCS-InstanceOptimalGeometricAlgorithms-AfshaniBarbayChan} or in higher dimensions~\cite{1985-SOCG-OutputSizeSensitiveAlgorithmsForFindingMaximalVectors-KirkpatrickSeidel}; while some others depend on the \emph{order} in which the input is given, such as for permutations~\cite{1992-ACJ-AnOverviewOfAdaptiveSorting-MoffatPetersson} but also for points in the plane~\cite{2002-SWAT-AdaptiveAlgorithmsForConstructingConvexHullsAndTriangulationsOfPolygonalChains-LevcopoulosLingasMitchell,2011-IEICE-AdaptiveAlgorithmsForPlanarConvexHullProblems-AhnOkamoto}.


Barbay~et al.~\cite{2016-ARXIV-SynergisticSortingAndDeferredDataStructuresOnMultiSets-BarbayOchoaSatty} described various ``synergistic'' solutions on multisets, which take advantage of both the input structure and the input order, in such a way that each of their solutions is never asymptotically slower than other solutions taking advantage of a subset of these features, but also so that on some family of instances, each of their solution performs an order of magnitude faster than any solution taking advantage of only a subset of those features. They left open the generalization of their results to higher dimensions.

In the context of the computation of the \textsc{Maxima Set} and of the \textsc{Convex Hull}, various refinements of the worst case complexity over instances of fixed size have been known for some time.
Kirkpatrick and Seidel described algorithms optimal in the worst case over instances of fixed input and output size, first in 1985 for the computation of the \textsc{Maxima Set} of points in any dimension~\cite{1985-SOCG-OutputSizeSensitiveAlgorithmsForFindingMaximalVectors-KirkpatrickSeidel} and then in 1986 for the computation of the \textsc{Convex Hull} in the plane~\cite{1986-JCom-TheUltimatePlanarConvexHullAlgorithm-KirkpatrickSeidel}: such results can be classified as focused on the \textbf{input structure}, and were further refined in 2009 by Afshani et al.~\cite{2009-FOCS-InstanceOptimalGeometricAlgorithms-AfshaniBarbayChan} in 2 and 3 dimensions.
Following a distinct approach, Levcopoulos et al.~\cite{2002-SWAT-AdaptiveAlgorithmsForConstructingConvexHullsAndTriangulationsOfPolygonalChains-LevcopoulosLingasMitchell} in 2002, and Ahn and Okamoto~\cite{2011-IEICE-AdaptiveAlgorithmsForPlanarConvexHullProblems-AhnOkamoto} in 2011, studied the computation of the \textsc{Convex Hull} in conjunction with various notions of \textbf{input order}: these results can be generalized to take advantage of the input order when computing \textsc{Maxima Sets} (Section~\ref{sec:inputOrderAdaptivePlanarMaxima}) and can be further refined using recent techniques (Section~\ref{sec:inputOrderAdaptiveConvexHull}).
Yet no algorithm (beyond a trivial dovetailing combination of the solutions described above) is known to take advantage of both the \textbf{input structure} and \textbf{input order} at the same time for the computation of the \textsc{Maxima Set} or of the \textsc{Convex Hull} of points, in the plane or in higher dimension, nor for any other problem than \textsc{Sorting Multisets}.

\paragraph{Hypothesis.}
It seems reasonable to expect that Barbay et al.'s synergistic results~\cite{2016-ARXIV-SynergisticSortingAndDeferredDataStructuresOnMultiSets-BarbayOchoaSatty} on \textsc{Sorting Multisets} should generalize to similar problems in higher dimension, such as the computation of the \textsc{Maxima Set} and of the \textsc{Convex Hull} of a set of points in the plane.
Yet these two problems present new difficulties of their own: 
(1) while the results on multisets~\cite{2016-ARXIV-SynergisticSortingAndDeferredDataStructuresOnMultiSets-BarbayOchoaSatty} are strongly based on a variant of Demaine et al.'s instance optimal algorithm to \textsc{Merge Multisets}~\cite{2000-SODA-AdaptiveSetIntersectionsUnionsAndDifferences-DemaineLopezOrtizMunro}, at this date no such results are known for \textsc{Merging Maxima Sets}, and the closest known result for \textsc{Merging Convex Hulls}~\cite{2008-CCCG-ConvexHullOfTheUnionOfConvexObjectsInThePlane-BarbayChen} (from 2008) is not adaptive to the size of the output, and hence not adaptive to the structure of the instance; furthermore
(2) whereas many \emph{input order adaptive} results are known for \textsc{Sorting Multisets} (with two surveys in 1992 on the topic~\cite{1992-ACJ-AnOverviewOfAdaptiveSorting-MoffatPetersson,1992-ACMCS-ASurveyOfAdaptiveSortingAlgorithms-EstivillCastroWood}, and various additional results~\cite{2009-Chapter-PartialSolutionAndEntropy-Takaoka,2013-TCS-OnCompressingPermutationsAndAdaptiveSorting-BarbayNavarro} since then), it seems that none are known for the computation of the \textsc{Maxima Set} and that only a few are known for the computation of the \textsc{Convex Hull}~\cite{2011-IEICE-AdaptiveAlgorithmsForPlanarConvexHullProblems-AhnOkamoto,2002-SWAT-AdaptiveAlgorithmsForConstructingConvexHullsAndTriangulationsOfPolygonalChains-LevcopoulosLingasMitchell}.

\paragraph{Our Results.}
\begin{LONG}After reviewing previous results on the computation of \textsc{Convex Hull} taking advantage of either the input structure (Section~\ref{sec:InputStructure}) or the input order (Section~\ref{sec:InputOrder}), and one result on \textsc{Sorting Multisets} taking advantage of both (Section~\ref{sec:synerg}), w\end{LONG}\begin{SHORT} W\end{SHORT}e confirm the hypothesis by
(1) presenting new solutions for \textsc{Merging Maxima Sets} (Section~\ref{sec:quick-maxima}) and \textsc{Merging Convex Hulls} (Section~\ref{sec:UpperHullUnion}) in the plane,
(2) defining new techniques to take advantage of the input order to compute \textsc{Maxima Sets} (Section~\ref{sec:inputOrderAdaptivePlanarMaxima}),
(3) improving previous techniques to analyze the computation of \textsc{Convex Hulls} in function of the input order (Section~\ref{sec:inputOrderAdaptiveConvexHull}), and
(4) synthesizing all those results in synergistic algorithms to compute \textsc{Maxima Sets} (Section \ref{sec:synergMaxima}) and \textsc{Convex Hulls} (Section~\ref{sec:synergisticUpperHull}) of a set of planar points.
For the sake of pedagogy, \textbf{we present our results incrementally, from the simplest to the most complex}.
\begin{LONG}
  We define formally the notions of input order, input structure and synergistic solution in Section~\ref{sec:back}.
\end{LONG}
We then describe synergistic solutions for the computation of both the \textsc{Maxima Set} (Section~\ref{sec:maxima}) and of the \textsc{Convex Hull} (Section~\ref{sec:convex}) of points in the plane (the latter requiring more advance techniques).
\begin{LONG}
In both case, our solution is based on an algorithm merging several partial solutions (Sections~\ref{sec:quick-maxima} and \ref{sec:UpperHullUnion}), adapted from Barbay et al.'s \texttt{Quick Hull Merge} algorithm~\cite{2016-ARXIV-SynergisticSortingAndDeferredDataStructuresOnMultiSets-BarbayOchoaSatty} to merge multisets\footnote{Itself inspired from Demaine et al.'s algorithm solving the same problem~\cite{2001-SODA-lowerboundindex-DemaineLopezOrtiz}.}.
  We conclude in Section~\ref{sec:discussion} with a partial list of issues left open for improvement.
\end{LONG}
Due to space constraints we only state our results in the article and defer all the proofs to the appendix.

\begin{BACK}
  \section{Background}
  \label{sec:back}

  Beyond the worst case complexity over instances of fixed size, the adaptive analysis of algorithms refines the scope of the analysis by considering the worst case complexity over finer classes of instances.
  \begin{SHORT}
    We describe here a selection of results adaptive to the \emph{input structure}, to the \emph{input order} and to both, \emph{synergistically}.
  \end{SHORT}
  \begin{VLONG}
    We describe here some relevant results along two axis: results about the computation of the \textsc{Convex Hull} which take advantage of the \emph{input structure}\begin{LONG} (Section~\ref{sec:InputStructure})\end{LONG}, results about the computation of the \textsc{Convex Hull} which take advantage of the \emph{input order}\begin{LONG} (Section~\ref{sec:InputOrder})\end{LONG}, and one result about \textsc{Sorting Multisets} which depends of both\begin{LONG} (Section~\ref{sec:synerg})\end{LONG}.
  \end{VLONG}

  \subsection{Input Structure}
  \label{sec:InputStructure}
  In 1985, Kirkpatrick and Seidel~\cite{1985-SOCG-OutputSizeSensitiveAlgorithmsForFindingMaximalVectors-KirkpatrickSeidel} described an algorithm to compute the \textsc{Maxima Set} of points in any dimension, which is optimal in the worst case over instances of input size $n$ and output size $h$, running in time within $O(n\log h)$.
  One year later, in 1986, they~\cite{1986-JCom-TheUltimatePlanarConvexHullAlgorithm-KirkpatrickSeidel} described a slightly more complex algorithm to compute the \textsc{Convex Hull} in the plane, which is similarly optimal in the worst case over instances of input size $n$ and output size $h$, running in time within $O(n\log h)$.
  Both results are described as \emph{output sensitive}, in the sense that the complexity depends on the size of the output, and can be classified as adaptive to the \emph{input structure}, as the position of the points clearly determine the output (and its size), as opposed to algorithms described in the next paragraph taking advantage of the order in which those points are given.

  \subsection{Input Order}
  \label{sec:InputOrder}

\begin{SHORT}
  In 2002, inspired by previous results on \textsc{Sorting Permutations}~\cite{1994-IC-SortingShuffledMonotoneSequences-LevcopoulosPetersson}, Levcopoulos et al.~\cite{2002-SWAT-AdaptiveAlgorithmsForConstructingConvexHullsAndTriangulationsOfPolygonalChains-LevcopoulosLingasMitchell} described an adaptive algorithm to compute the \textsc{Convex Hull} of $n$ points in the plane in time within $O(n\lg\kappa)$, where $\kappa$ is the minimal number of simple chains into which the input sequence of points can be partitioned.
\end{SHORT}
\begin{LONG}
  A \emph{polygonal chain} is a curve specified by a sequence of points $p_1, \dots, p_n$. The curve itself consists of the line segments connecting the pairs of consecutive points. A polygonal chain $C$ is \emph{simple} if any two edges of $C$ that are not adjacent are disjoint, or if the intersection point is a vertex of $C$; and any two adjacent edges share only their common vertex. Melkman~\cite{1987-IPL-OnLineConstructionOfTheConvexHullOfASimplePolyline-Melkman} described an algorithm that computes the {\sc{Convex Hull}} of a simple polygonal chain in linear time, and Chazelle~\cite{1991-DCG-TriangulatingASimplePolygonInLinearTime-Chazelle} described an algorithm for testing whether a polygonal chain is simple in linear time.
  In 2002, Levcopoulos et al.~\cite{2002-SWAT-AdaptiveAlgorithmsForConstructingConvexHullsAndTriangulationsOfPolygonalChains-LevcopoulosLingasMitchell} combined these results to yield an algorithm for computing the {\sc{Convex Hull}} of polygonal chains.
  Their algorithm tests if the chain $C$ is simple, using Chazelle's algorithm~\cite{1991-DCG-TriangulatingASimplePolygonInLinearTime-Chazelle}: if the chain $C$ is simple, the algorithm computes the {\sc{Convex Hull}} of $C$ in linear time, using Melkman's algorithm~\cite{1987-IPL-OnLineConstructionOfTheConvexHullOfASimplePolyline-Melkman}. Otherwise, if $C$ is not simple, the algorithm partitions $C$ into the subsequences $C'$ and $C''$, whose sizes differ at most in one; recurses on each of them; and merges the resulting {\sc{Convex Hulls}} using Preparata and Shamos's algorithm~\cite{1985-BOOK-ComputationalGeometryAnIntroduction-PreparataShamos}.
  They measured the complexity of this algorithm in terms of the minimum number of simple subchains $\kappa$ into which the chain $C$ can be partitioned. Let $t(n, \kappa)$ be the worst-case time complexity taken by this algorithm for an input chain of $n$ vertices that can be partitioned into $\kappa$ simple subchains.  They showed that $t(n, \kappa)$ satisfies the following recursion relation: $t(n, \kappa) \leq t(\lceil \frac{n}{2} \rceil, \kappa_1) + t(\lfloor \frac{n}{2} \rfloor, \kappa_2), \kappa_1 + \kappa_2 \leq \kappa + 1$. The solution to this recursion gives $t(n, \kappa) \in O(n(1{+}\log{\kappa}))\subseteq O(n\log n)$.
\end{LONG}
\begin{SHORT}
  In 2011, Ahn and Okamoto~\cite{2011-IEICE-AdaptiveAlgorithmsForPlanarConvexHullProblems-AhnOkamoto} considered two variants of the problem, where the output is a permutation of the input such that the \textsc{Convex Hull} can be checked and extracted in linear time from it. They describe an adaptive result directly inspired from a disorder measure introduced for the study of adaptive algorithms for \textsc{Sorting Permutations}, the number \texttt{Inv} of inversions in the permutation~\cite{1992-ACJ-AnOverviewOfAdaptiveSorting-MoffatPetersson,1992-ACMCS-ASurveyOfAdaptiveSortingAlgorithms-EstivillCastroWood}.
\end{SHORT}
\begin{LONG}
  
  In 2011, Ahn and Okamoto~\cite{2011-IEICE-AdaptiveAlgorithmsForPlanarConvexHullProblems-AhnOkamoto} followed a distinct approach for the computation of the \textsc{Convex Hull}, also based on some notions of input order. They considered a variant of the problem where the output is the same size of the input, but such that the \textsc{Convex Hull} can be checked and extracted in linear time from this output. In this context, they describe adaptive results directly inspired from disorder measures introduced for the study of adaptive algorithms for \textsc{Sorting Permutations}, such as \texttt{Runs} and \texttt{Inv}~\cite{1992-ACJ-AnOverviewOfAdaptiveSorting-MoffatPetersson,1992-ACMCS-ASurveyOfAdaptiveSortingAlgorithms-EstivillCastroWood}.
\end{LONG}

Inspired by Ahn and Okamoto's definition~\cite{2011-IEICE-AdaptiveAlgorithmsForPlanarConvexHullProblems-AhnOkamoto}, we define some simple measure of \emph{input order} for the computation of \textsc{Maxima Sets} in Section~\ref{sec:inputOrderAdaptivePlanarMaxima}, and we slightly refine Levcopoulos et al.'s analysis~\cite{2002-SWAT-AdaptiveAlgorithmsForConstructingConvexHullsAndTriangulationsOfPolygonalChains-LevcopoulosLingasMitchell} for the computation of \textsc{Convex Hulls} in Section~\ref{sec:inputOrderAdaptiveConvexHull}.

\subsection{Synergistic Solutions}
\label{sec:synerg}

Inspired by previous results on sorting multisets in a way adaptive to the frequencies of the element~\cite{1976-JComp-SortingAndSearchingInMultisets-MunroSpira} on one hand, and on sorting permutation in a way adaptive to the distribution of the lengths of the subsequences of consecutive positions already sorted~\cite{2009-Chapter-PartialSolutionAndEntropy-Takaoka} on the other hand, Barbay~et al.~\cite{2016-ARXIV-SynergisticSortingAndDeferredDataStructuresOnMultiSets-BarbayOchoaSatty} described two ``synergistic'' algorithms \textsc{Sorting Multisets}, which take advantage both of the input structure and of the input order, in such a way that each of their solutions is never asymptotically slower than other solutions taking advantage of a subset of these features, but also so that on some family of instances, each of their solution performs an order of magnitude faster than any solution taking advantage of only a subset of those features. They left open the generalization of their results to higher dimensions. We generalize their results to dimension 2, for the computation of the \textsc{Maxima Set} in Section~\ref{sec:maxima}, and for the computation of the \textsc{Convex Hull} in Section~\ref{sec:convex}.
\end{BACK}

\section{Maxima Set}
\label{sec:maxima}

Given a point $p\in\mathbb{R}^2$, let $p_x$ and $p_y$ denote the $x$- and $y$-coordinates of $p$, respectively. Given two points $p$ and $q$, $p$ \emph{dominates} $q$ if $p_x \ge q_x$ and $p_y \ge q_y$.  Given a set $\mathcal{S}$ of points in $d$ dimensions, a point $p$ from $\mathcal{S}$ is called \emph{maximal} if none of the other points of $\mathcal{S}$ dominates $p$.  The \textsc{Maxima Set} of such a set $\mathcal{S}$ is the uniquely defined set of all maximal points~\cite{1975-JACM-OnFindingTheMaximaOfASetOfVectors-KungLuccioPreparata}.  Kirkpatrick and Seidel~\cite{1985-SOCG-OutputSizeSensitiveAlgorithmsForFindingMaximalVectors-KirkpatrickSeidel} described an algorithm that computes the \textsc{Maxima Set} running in time within $O(n\log h)$, where $n$ is the number of input points, and $h$ is the number of points in the \textsc{Maxima set} (i.e., the size of the output).
In 2009, Afshani et al.~\cite{2009-FOCS-InstanceOptimalGeometricAlgorithms-AfshaniBarbayChan} improved the results on the computation of both the \textsc{Maxima Set} and \textsc{Convex Hull}
\begin{LONG}
  in dimension 2 and 3
\end{LONG}
by taking the best advantage of the relative positions of the points (while ignoring the input order).
\begin{INUTILE}
  describing ``input order oblivious instance optimal'' algorithms which, for any instance $I$, performs within a constant factor of the performance of any algorithm that does not take advantage of the order of the points in the algebraic decision tree model.
\end{INUTILE}

\begin{TODO}
CHECK that I am correctly referencing the computational model.
\end{TODO}

If the points are sorted by their coordinates (say, in lexicographic order of their coordinates for a fixed order of the dimensions), the \textsc{Maxima Set} can be computed in time linear in the size of the input. Refining this insight, we show in Section~\ref{sec:inputOrderAdaptivePlanarMaxima} that one can take advantage of the input order even if it is not as strictly sorted: this presents a result which is orthogonal to previous input structure adaptive results~\cite{1985-SOCG-OutputSizeSensitiveAlgorithmsForFindingMaximalVectors-KirkpatrickSeidel,2009-FOCS-InstanceOptimalGeometricAlgorithms-AfshaniBarbayChan}.  In order to combine these results synergistically with previous input structure adaptive results~\cite{1985-SOCG-OutputSizeSensitiveAlgorithmsForFindingMaximalVectors-KirkpatrickSeidel,2009-FOCS-InstanceOptimalGeometricAlgorithms-AfshaniBarbayChan}, we study in Section~\ref{sec:quick-maxima} an algorithm that solves the problem of \textsc{Merging Maxima}, which asks for computing the \textsc{Maxima Set} of the union of maxima sequences, in such a way that it outperforms both input structure adaptive results by taking advantage of the number and sizes of the maxima sequences and of the relative positions between the points in them. Last, we combine those results into a single synergistic algorithm, which decomposes the input sequence of points into several ``easy'' subsequences for which the corresponding \textsc{Maxima Set} can be computed in linear time, and then proceeds to merge them.  The resulting algorithm not only outperforms previous input structure adaptive results~\cite{1985-SOCG-OutputSizeSensitiveAlgorithmsForFindingMaximalVectors-KirkpatrickSeidel,2009-FOCS-InstanceOptimalGeometricAlgorithms-AfshaniBarbayChan} in the sense that it never performs (asymptotically) worse, and performs better when it can take advantage of the input order, it also outperforms a dovetailing combination of previous input structure adaptive algorithms~\cite{1985-SOCG-OutputSizeSensitiveAlgorithmsForFindingMaximalVectors-KirkpatrickSeidel,2009-FOCS-InstanceOptimalGeometricAlgorithms-AfshaniBarbayChan} and the input order adaptive algorithm described in Section~\ref{sec:inputOrderAdaptivePlanarMaxima}.

\subsection{Input Order Adaptive  Maxima Set}
\label{sec:inputOrderAdaptivePlanarMaxima}

In many cases the \textsc{Maxima Set} can be computed in time linear in the size of the input, independently from the size of the \textsc{Maxima Set} itself.
For instance, consider an order of the input where
(1) the maximal points are given in order sorted by one coordinates, and
(2) for each maximal point $p$, all the points dominated by $p$ are given immediately after $p$ in the input order (in any relative order). 
The \textsc{Maxima Set} of a sequence of points given in this order can be extracted and validated in linear time by a simple greedy algorithm, which throws an exception if the input is not in such an order.
Each of the various ways to deal with such exceptions directly yields an input order adaptive algorithm~\cite{2011-IEICE-AdaptiveAlgorithmsForPlanarConvexHullProblems-AhnOkamoto}.
\begin{LONG}
  For instance, if the point found to be out of order is \emph{inserted} in the partial \textsc{Maxima Set} computed up to this point, this yields an algorithm running in time within $O(n\lg\mathtt{Inv})$ where $\mathtt{Inv}$ is the sum of such insertion costs.
\end{LONG}

Let's label such a sequence ``\emph{smooth}'', and by extension any input subsequence of consecutive positions which have the same property.
Given an input sequence $\mathcal{S}$, let $\sigma$ denote
\begin{INUTILE}
  let's call its \emph{smoothness}
\end{INUTILE}
the minimal number of smooth subsequences into which it can be decomposed.
\begin{INUTILE}
  , and denote it by $\sigma$
\end{INUTILE}
Most interestingly for synergistic purpose, such a decomposition can be computed in time linear in the input size.
Detecting such $\sigma$ smooth subsequences and merging them two by two yields an algorithm running in time within $O(n(1+\log\sigma))$.  Such a result is orthogonal to previous input structure adaptive results~\cite{1985-SOCG-OutputSizeSensitiveAlgorithmsForFindingMaximalVectors-KirkpatrickSeidel, 2009-FOCS-InstanceOptimalGeometricAlgorithms-AfshaniBarbayChan}: it can be worse than $O(n\log h)$ when the output size $h$ is small and the input is in a ``bad'' order, as it can be much better than $O(n\log h)$ when $h$ is large and the input is in a ``good'' order. We show in the next two sections an algorithm which outperforms both.

\subsection{Union of Maxima Sequences}
\label{sec:quick-maxima}

We describe the \texttt{Quick Union Maxima} algorithm, which computes the \textsc{Maxima Set} of the union of maxima sequences in the plane, assuming that the points in the maxima sequences are given in sorted order by their $x$-coordinates (i.e., \textsc{Merging Maxima}). This algorithm generalizes Barbay et al.'s~\cite{2016-ARXIV-SynergisticSortingAndDeferredDataStructuresOnMultiSets-BarbayOchoaSatty} \texttt{QuickSort} inspired algorithm for \textsc{Merging Multiset} and is a building block towards the synergistic algorithm for computing the \textsc{Maxima Set} of a set of planar points described in Section~\ref{sec:synergMaxima}. Given a maxima sequence $\mathcal{M}_i$, let $\mathcal{M}_i[a]$ and $\mathcal{M}_i[b..c]$ denote the $a$-th point and the block of consecutive $c-b+1$ points corresponding to positions from $b$ to $c$ in $\mathcal{M}_i$, respectively.
\begin{LONG}
  As its name indicates, the algorithm is inspired by the \texttt{QuickSort} algorithm.
\end{LONG}

\subsubsection{Description of the algorithm Quick Union Maxima.}
The \texttt{Quick Union Maxima} algorithm chooses a point $p$ that forms part of the \textsc{Maxima Set} of the union, and discards all the points dominated by $p$.
\begin{LONG}
  Note that all dominated points do not belong to the maxima that contains $p$.
\end{LONG}
The selection of $p$ ensures that at least half of the maxima sequences will have points dominated by $p$ or to the right of $p$ and at least half of the maxima sequences will have points dominated by $p$ or to the left of $p$. The algorithm identifies a block $\mathcal{B}$ of consecutive points in the maxima sequence that contains $p$, which forms part of the output \textsc{Maxima Set} ($p$ is contained in $\mathcal{B}$).  All the points in $\mathcal{B}$ are discarded.
\begin{LONG}
  Note that if the points of a maxima sequence in the plane are sorted in ascending order by their $x$-coordinates, then their $y$-coordinates are sorted in decreasing order. This algorithm takes advantage of this fact when it discards points.
\end{LONG}
All discarded points are identified by doubling searches~\cite{1976-IPL-AnAlmostOptimalAlgorithmForUnboundedSearching-BentleyYao}\begin{LONG}
  inside the maxima sequences
\end{LONG}. The algorithm then recurses separately on the non-discarded points to the left of $p$ and on the non-discarded points to the right of $p$. (See Algorithm~\ref{alg:qum} for a more formal description.)
Next, we analyze the time complexity of the \texttt{Quick Union Maxima} algorithm.

\begin{algorithm} 
  \caption{\texttt{Quick Union Maxima}} 
  \label{alg:qum} 
  \begin{algorithmic}[1] 

    \REQUIRE{A set $\mathcal{M}_1, \dots, \mathcal{M}_\rho$ of $\rho$ maxima sequences}
    \ENSURE{The \textsc{Maxima Set} of the union of $\mathcal{M}_1, \dots, \mathcal{M}_\rho$}
    
    \STATE Compute the median $\mu$ of the $x$-coordinates of the
    middle points of the maxima sequences;
    
    \STATE Perform doubling searches for the value $\mu$ in the
    $x$-coordinates of the points of all maxima sequences,
    starting at both ends of the maxima sequences in parallel;
    
    \STATE Find the point $p$ of maximum $y$-coordinate among the points
    $q$ such as $q_x \ge \mu$ in all maxima sequences, note $j\in[1..\rho]$
    the index of the maxima sequence containing $p$;
    
    \STATE Discard all points dominated by $p$ through doubling
    searches for the values of $p_x$ and $p_y$ in the $x$- and $y$-coordinates
    of all maxima, respectively, except $\mathcal{M}_j$ (Search for $p_x$ in the
    points $q$ such that $q_x \ge \mu$ and for $p_y$ in the points $q$ such that $q_x<\mu$);
    
    \STATE Find the point $r$ of maximum $y$-coordinate among the
    points $q$ such that $q_x > p_x$ in all maxima sequences except $\mathcal{M}_j$
    and find the point $\ell$ of maximum $x$-coordinate among the points $q$ such that
    $q_y > p_y$ in all maxima sequences except $\mathcal{M}_j$;

    \STATE Discard the block in $\mathcal{M}_j$ containing $p$ that forms part of the
    output through doubling searches for the values of $\ell_x$ and $r_y$ in the $x$- and
    $y$-coordinates of the points in $\mathcal{M}_j$, respectively. (Search for $\ell_x$ in the points
    $q$ such that $q_x < p_x$ and for $r_y$ in the points $q$ such that $q_x > p_x$.);
    
    \STATE Recurse separately on the non-discarded points left and right of $p$.

  \end{algorithmic}
\end{algorithm}

\subsubsection{Complexity Analysis of Quick Union Maxima.}
\label{sec:analysis-qum}


Every algorithm for \textsc{Merging Maxima} needs to certify that blocks of consecutive points in the maxima sequences are dominated or are in the \textsc{Maxima Set} of the union. In the following we formalize the notion of a \emph{certificate}, that permits to check the correctness of the output in less time than to recompute the output itself. We define a ``language'' of basic ``arguments'' for such certificates: \emph{domination} (which discards points from the input) and \emph{maximality} (which justify the presence of points in the output) arguments, and their key positions in the instance. A certificate will be verified by checking each of its arguments: those can be checked in time proportional to the number of blocks in them.

\begin{LONG}
  \begin{definition}
    $\langle \mathcal{M}_i[a] \supset \mathcal{M}_j[b..c] \rangle$
    is an \emph{elementary domination argument} if the point
    $\mathcal{M}_i[a]$ dominates all the points in the block
    $\mathcal{M}_j[b..c]$.
  \end{definition}
\end{LONG}

\begin{definition}
  $\langle \mathcal{M}_i[a] \supset \mathcal{M}_{j_1}[b_1..c_1],
  \dots, \mathcal{M}_{j_t}[b_t..c_t] \rangle$ is a \emph{Domination
    Argument} if the point $\mathcal{M}_i[a]$ dominates all the points
  in the blocks
  $\mathcal{M}_{j_1}[b_1..c_1], \dots,
  \mathcal{M}_{j_t}[b_t..c_t]$.
\end{definition}

\begin{LONG}
  \begin{lemma}
    A \emph{domination argument} $\langle \mathcal{M}_i[a] \supset \mathcal{M}_{j_1}[b_1..c_1], \dots, \mathcal{M}_{j_t}[b_t..c_t] \rangle$ can be checked in $O(t)$ data comparisons.
  \end{lemma}
\end{LONG}

It is not enough to eliminate all points that can not participate in the output. Certifying would still require additional work: a correct algorithm must justify the optimality of its output. To this end we define maximality arguments.

\begin{definition}
  $\langle \mathcal{M}_i[a..b] \dashv \mathcal{M}_{j_1}[a_1..b_1],
  \dots, \mathcal{M}_{j_t}[a_t..b_t] \rangle$ is a \emph{Maximality
    Argument} if \textbf{either} the points $\mathcal{M}_i[b]$
  dominates the points
  $\mathcal{M}_{j_1}[a_1], \dots, \mathcal{M}_{j_t}[a_t]$ and
  the $x$-coordinates of the points
  $\mathcal{M}_{j_1}[a_1-1], \dots, \mathcal{M}_{j_t}[a_t-1]$ are
  less than the $x$-coordinate of the point ${\cal M}_i[a]$
  \textbf{or} the point $\mathcal{M}_i[a]$ dominates the points
  $\mathcal{M}_{j_1}[b_1], \dots, \mathcal{M}_{j_t}[b_t]$ and
  the $y$-coordinates of the points
  $\mathcal{M}_{j_1}[b_1+1], \dots, \mathcal{M}_{j_t}[b_t+1]$ are
  less than the $y$-coordinate of the point ${\cal M}_i[b]$.
\end{definition}

If $\langle \mathcal{M}_i[a..b] \dashv \mathcal{M}_{j_1}[a_1..b_1], \dots, \mathcal{M}_{j_t}[a_t..b_t] \rangle$ is a valid \emph{maximality argument}, then the points in the block $\mathcal{M}_i[a..b]$ are maximal among the maxima sequences $\mathcal{M}_i, \mathcal{M}_{j_1}, \dots, \mathcal{M}_{j_t}$.
\begin{LONG}
  \begin{lemma}
    A \emph{maximality argument} $\langle \mathcal{M}_i[a..b] \dashv \mathcal{M}_{k_1}[a_1..b_1], \dots, \mathcal{M}_{k_t}[a_t..b_t] \rangle$ can be checked in $O(t)$ data comparisons.
  \end{lemma}
\end{LONG}
The difficulty of finding and describing domination and maximality arguments depend on the points they refer to in the maxima sequences, a notion captured by ``argument points'':

\begin{definition}
  Given an argument
  ${\cal A} = \langle \mathcal{M}_i[a] \supset
  \mathcal{M}_{j_1}[b_1..c_1], \dots, \mathcal{M}_{j_t}[b_t..c_t]
  \rangle$ or
  ${\cal B} = \langle \mathcal{M}_i[a..b] \dashv
  \mathcal{M}_{j_1}[a_1..b_1], \dots, \mathcal{M}_{j_t}[a_t..b_t]
  \rangle$, the \emph{Argument Points} are the points $\mathcal{M}_i[a]$ in
  ${\cal A}$ and $\mathcal{M}_i[a]$ and $\mathcal{M}_i[b]$ in
  ${\cal B}$.
\end{definition}

Those atomic arguments combine into a general definition of a certificate that any correct algorithm for \textsc{Merging Maxima} in the comparison model can be modified to output.

\begin{definition}
  Given a set of maxima sequences and their \textsc{Maxima Set} $\mathcal{M}$ expressed as several blocks on the maxima sequences. A \emph{certificate} of $\mathcal{M}$ is a set of domination and maximality arguments such that the \textsc{Maxima Set} of any instance satisfying those arguments is given by the description of $\mathcal{M}$. The length of a certificate is the number of distinct argument points in it.
\end{definition}

\begin{INUTILE}
  \begin{definition}
    Given an instance ${\cal I}$ of the \textsc{Union Maxima} problem,
    a certificate ${\cal C}$ for ${\cal I}$ of \emph{minimal length} is
    a certificate with the minimal number of distinct argument points.
  \end{definition}
\end{INUTILE}

\begin{LONG}
  We divide the analysis of the time complexity of the \texttt{Quick Union Maxima} algorithm into two lemmas.
\end{LONG}
\begin{LONG}
  We first bound the cumulated time complexity of the doubling searches for the value of the median $\mu$ of the $x$-coordinates of the middle points of the maxima (i.e., step $2$ of Algorithm~\ref{alg:qum}) and the doubling searches in the points discard steps (i.e., steps $4$ and $6$ of Algorithm~\ref{alg:qum}).
\end{LONG}
The algorithm partitions the maxima sequences into blocks of consecutive discarded points, where each block is discarded because it is dominated, or because this block forms part of the \textsc{Maxima Set} of the union. Each block forms part of some argument of the certificate computed by the algorithm.
\begin{SHORT}
  This fact is key to bound the cumulated time complexity of the steps $2, 4$ and $6$ of Algorithm~\ref{alg:qum}. The key fact to bound the time complexity of the steps $1,3$ and $5$ of Algorithm~\ref{alg:qum} is that each execution of these steps has time complexity bounded by the number of maxima sequences in the subinstance.\end{SHORT}
\begin{VLONG}
  \begin{lemma}\label{lem:blocks}
    Let $s_{1}, \dots, s_{\beta}$ be the sizes of the $\beta$ blocks into which the algorithm \texttt{Quick Union Maxima} divides the maxima sequence ${\cal M}_i$. The cumulated time complexity of the doubling searches for the value of the medians $\mu$ of the $x$-coordinates of the middle points of the maxima sequences (i.e., step $2$) and the doubling searches in the points discard steps (i.e., steps $4$ and $6$) of the algorithm \texttt{Quick Union Maxima} in the maxima sequence ${\cal M}_i$ is within $O(\sum_{j=1}^{\beta}\log{s_{j}})$.
  \end{lemma}

  \begin{PROOF}
    \begin{proof}
        Every time the algorithm finds the insertion rank of one of the medians $\mu$ of the $x$-coordinates of the middle points of the maxima sequences in $\mathcal{M}_i$, it finds a position $d$ inside a blocks whose points will be discarded. The discard points steps that search for the insertion rank of $p_x$ and $p_y$ start the search in the position $d$. The time complexity of both discard points steps are bounded by $O(\log{s_b})$, where $s_b$ is the size of the discarded block $b$. Both discard points steps partition $\mathcal{M}_i$ in positions separating the blocks to the left of $b$, the block $b$ itself and the blocks to the right of $b$.

The combination of the doubling search that finds the insertion rank of $\mu$ in the $x$-coordinates of $\mathcal{M}_i$ with the doubling searches that discard points
\begin{LONG}
  starting in $d$
\end{LONG}
can be represented as a tree. Each internal node has two children, which correspond to the two subproblems into which the recursive steps partition $\mathcal{M}_i$\begin{LONG}
  , the blocks to the left of $b$ and the blocks to the right of $b$
\end{LONG}. The cost of this combination is bounded by $O(\log s_b + \log s)$, where $s$ is the minimum between the sum of the sizes of the blocks to the left of $b$ and the sum of the sizes of the blocks to the right of $b$, because of the two doubling searches in parallel. The size of each internal node is the size of the block discarded in this step. The size of each leaf is the sum of the sizes of the blocks in the child subproblem represented by this leaf.

We prove that at each combination of steps, the total cost is bounded by eight times of the sum of the logarithms of the sizes of the nodes in the tree. This is done by induction over the number of steps. If the number of steps is zero then there is no cost. For the inductive step, if the number of steps increases by one, then a new combination of steps is done and a leaf subproblem is partitioned into two new subproblems. At this step, a leaf of the tree is transformed into an internal node and two new leaves are created. Let $w$ and $z$ such that $w \leq z$ be the sizes of the new leaves created. Note that $w$ and $z$ are the sum of the sizes of the blocks to the left and to the right, respectively, of the discarded block $b$ in this step. The cost of this step is less than $4\log{w} + 4\log{b}$. The cost of all the steps then increases by $4\log{w} + 4\log{b}$, and hence eight times the sum of the logarithms of the nodes in the tree increases by $8(\log{w} + \log{z} + \log b - \log({w+z+b}))$. But if $w \ge 3$, $w \ge b$ and $w \le z$ then the result follows. \qed
    \end{proof}
  \end{PROOF}
\end{VLONG}
\begin{INUTILE}
  The algorithm uses the points $p$, $\ell$ and $r$ to discard points
  from the input maxima. The discarded points could be part of the
  output or dominated points. We call the points $p$, $\ell$ and $r$
  \emph{argument points}.

  The number of argument points that the algorithm \texttt{Quick Union
    Maxima} uses to discard points is asymptotically the minimum
  number of argument points that any other algorithm that computes the
  union of $\rho$ maxima needs to use to discard points from the input
  maxima.

  We formalize this notion by comparing the number of \emph{argument
    points} that the \texttt{Quick Union Maxima} uses to discard
  points with an algorithm that computes the minimum number needed of
  \emph{argument points} when it computes the union of $\rho$ maxima.
\end{INUTILE}
\begin{LONG}
  We bound next the time complexity of the steps that compute the median $\mu$ of the $x$-coordinates of the middles points in the maxima (i.e., step $1$ of Algorithm~\ref{alg:qum}) and the steps that find the points $p$, $\ell$ and $r$ (i.e., steps $3$ and $5$ of Algorithm~\ref{alg:qum}) in the \texttt{Quick Union Maxima} algorithm.
  Note that one execution of these steps has time complexity bounded by the number of maxima sequences in the sub-instance.
\end{LONG}The partition of the maxima sequences by the $x$-coordinate of $p$ and the discarded points decrease the number of maxima sequences in the subinstances.

\begin{LONG}
  \begin{lemma}\label{lem:sequences}
    The cumulated number of comparison performed by the steps that compute the median $\mu$ of the $x$-coordinates of the middles points in the $\rho$ maxima sequences (i.e., step $1$) and the steps that find the points $p$, $\ell$ and $r$ (i.e., steps $3$ and $5$) in the \texttt{Quick Union Maxima} algorithm is within $O(\sum^{\delta}_{i=1}\log{\binom{\rho}{m_i}})$, where $\delta$ is the length of the certificate ${\cal C}$ computed by the algorithm and $m_1, \dots, m_\delta$ is a sequence where $m_i$ is the number of maxima sequences whose blocks form the $i$-th argument of ${\cal C}$.
  \end{lemma}
  
\begin{PROOF}
  \begin{proof}
    
    We prove this lemma by induction over $\delta$ and $\rho$. The time complexity of one of these steps is linear in the number of maxima sequences in the sub-instance (i.e., ignoring all the empty maxima sequences of this sub-instance).

    Let $\mathcal{T}(\delta,k)$ be the cumulated time complexity during the execution of the steps that compute the medians $\mu$ of the $x$-coordinates of the middles points (i.e., step $1$) and during the steps that find the points $p$, $\ell$ and $r$ (i.e., steps $3$ and $5$) in the algorithm \texttt{Quick Union Maxima}.
  We prove that $\mathcal{T}(\delta, k) \le \sum^{\delta}_{i=1}m_i\log{\frac{k}{m_i}} - k$, where $m_i$ is the number of maxima sequences whose blocks form the $i$-th argument of ${\cal C}$.
  Let $\mu$ be the first median of the $x$-coordinates of the middles points of the maxima sequences computed by the algorithm. Let $c$ and $d$ be the number of maxima sequences that have non-discarded point only above of $p_y$ and to the right of $p_x$, respectively. Let $b$ be the number of maxima sequences that have non-discarded points above $p_y$ and to the right of $p_x$. Let $e$ be the number of maxima sequences which all their points are dominated by $p$. Let $\delta_c$ and $\delta_d$ be the number of arguments computed by the algorithm to discard points in the maxima sequences above $p_y$ and to the right of $p_x$, respectively. Then, $\mathcal{T}(\delta, k) = \mathcal{T}(\delta_c, c+b) + \mathcal{T}(\delta_d, d+b) + k$ because of the two recursive calls and the steps $1$, $3$ and $5$ of the algorithm \texttt{Quick Union Maxima}. By Induction Hypothesis, $\mathcal{T}(\delta_c,c+b) \le \sum^{\delta_c}_{i=1}m_i\log{\frac{c+b}{m_i}} - c - b$ and $\mathcal{T}(\delta_d, d+b) \le \sum^{\delta_d}_{i=1}m_i\log{\frac{d+b}{m_i}} - d - b$. In the worst case $e=0$, and we need to prove that $c+d \le \sum^{\delta_c}_{i=1}m_i \log\left({1 + \frac{d}{c+b}}\right) + \sum^{\delta_d}_{i=1}m_i\left({1 + \frac{c}{d+b}}\right)$, but this is a consequence of $\sum^{\delta_c}_{i=1}m_i \ge c+b, \sum^{\delta_d}_{i=1}m_i \ge d+b$ (the number of discarded blocks is greater than or equal to the number of maxima sequences); $c \le d+b, d \le c + b$ (at least $\frac{k}{2}$ maxima sequences are left to the left and to the right of $\mu$); and $\log\left({1 + \frac{y}{x}}\right)^x \ge y$ for $y \le x$.  \qed
\end{proof}
\end{PROOF}

  Combining Lemma~\ref{lem:blocks} and Lemma~\ref{lem:sequences} yields an upper bound on the number of data comparisons performed by the algorithm \texttt{Quick Union Maxima}:
\end{LONG}

\begin{theorem}\label{theo:qum}
  Given $\rho$ maxima sequences. The \texttt{Quick Union Maxima} algorithm performs within $O(\sum_{j=1}^{\beta}\log{s_{j}} + \sum^{\delta}_{i=1}\log{\binom{\rho}{m_i}})$ data comparisons when it computes the \textsc{Maxima Set} of the union of these maxima sequences; where $\beta$ is the number of blocks in the certificate ${\cal C}$ computed by the algorithm; $s_1, \dots, s_\beta$ are the sizes of these blocks; $\delta$ is the length of ${\cal C}$; and $m_1, \dots, m_\delta$ is a sequence where $m_i$ is the number of maxima sequences whose blocks form the $i$-th argument of ${\cal C}$.
  \begin{LONG}
    This number of comparisons is optimal in the worst case over instances formed by $\rho$ maxima sequences that have certificates ${\cal C}$ formed by $\beta$ blocks of sizes $s_1, \dots, s_\beta$ and length $\delta$ such that $m_1, \dots, m_\delta$ is a sequence where $m_i$ is the number of maxima sequences whose blocks form the $i$-th argument of ${\cal C}$.
  \end{LONG}
\end{theorem}

The optimality of this algorithm is a consequence of the fact that it checks each arguments of any certificate using a constant number of argument points.
\begin{VLONG}
 We prove next that the \texttt{Quick Union Maxima} algorithm computes a \emph{certificate} which length is a constant factor of the length of the certificate of minimal length. Consider the following algorithm for \textsc{Merging Maxima}. The \texttt{Left-to-Right} algorithm chooses the first points from left to right of each maxima sequence and computes the points $u$ and $v$ of maximum and second maximum $y$-coordinate among these points, respectively. Let $\mathcal{M}_i$ be the maxima sequence that contains $u$. Let $a$ be the index of $u$ in $\mathcal{M}_i$. The $y$-coordinates of the points in the input maxima are sorted in decreasing order from left to right. The algorithm searches then for the insertion rank of $v_y$ in $\mathcal{M}_i$. Let $b$ be the index of the rightmost point $g$ in $\mathcal{M}_i$ such that $g_y > v_y$. The block $\mathcal{M}_i[a..b]$ form part of the \textsc{Maxima Set} of the union and are discarded by the \texttt{Left-to-Right} algorithm. If $g$ dominates $v$, the algorithm discards all points in the input dominated by $g$. The algorithm restarts the computation on the non-discarded points.

\begin{lemma}
The \texttt{Left-to-Right} algorithm computes a certificate of minimal length when it computes the union of $\rho$ maxima sequences.
\end{lemma}

\begin{PROOF}
\begin{proof}
  Let $u$ and $v$ be the points with maximum and second maximum $y$-coordinate among the first points from left to right of the non-discarded points of the maxima sequences. Let $\mathcal{M}_i$ be the maxima sequence that contains $u$. Then, all points in $\mathcal{M}_i$ with $y$-coordinate greater than $v_y$ could be discarded (i.e., form part of the output) and $v$ is the point in the maxima sequences that allows to discard the greatest number of consecutive points including $u$ in $\mathcal{M}_i$. Let $g$ and $h$ be consecutive points in $\mathcal{M}_i$ such that $g_y > v_y > h_y$. If $g$ dominates $v$, then $g$ is the rightmost point in $\mathcal{M}_i$ that dominates $v$. Hence, $g$ is the point in $\mathcal{M}_i$ that dominates the maximum number of consecutive points including $v$ in the maxima sequence that contains $v$. These two arguments are enough to prove that the algorithm computes the minimum number of \emph{argument points}.\qed
\end{proof}
\end{PROOF}

Let $a$ be the number of \emph{argument points} of the certificate computed by the \texttt{Left-to-Right} algorithm. The number of \emph{argument points} of the certificate that the algorithm \texttt{Quick Union Maxima} computes is within a constant factor of $a$.
\end{VLONG}

\begin{lemma}\label{lem:opt-max}
  The algorithm \texttt{Quick Union Maxima} computes a certificate
  which length is a constant factor of the length of the certificate of minimal length.
\end{lemma}

\begin{PROOF}
  \begin{proof}

      Suppose that there is a block $b$ of consecutive points in a maxima sequence that the \texttt{Left-to-Right} algorithm discards because it identifies that the points in $b$ are in the output. Suppose that the algorithm \texttt{Quick Union Maxima} running in the same input computes a point $p$ ($p$ is the point of maximum $y$-coordinate among the points of $x$-coordinate greater than $\mu$, then $p$ is in the output) that is contained in $b$. Let $r$ be the point of maximum $y$-coordinate among the points with $x$-coordinate greater than $p_x$ computed by the \texttt{Quick Union Algorithm}. Let $h$ be the \emph{argument point} used by the \texttt{Left-to-Right} algorithm to identify the rightmost point in $b$. Hence, $r_y < h_y$. Let $\ell$ be the point of maximum $x$-coordinate among the points with $y$-coordinate greater than $p_y$ computed by the \texttt{Quick Union Maxima} algorithm. Let $u$ be the \emph{argument point} used by the \texttt{Left-to-Right} algorithm to discard dominated points before the identification of $b$. Hence, $\ell$ is the same point as $u$. So, the algorithm \texttt{Quick Union Maxima} discards at least the block $b$ using a constant number of \emph{argument points}. The result follows.\qed
  \end{proof}
\end{PROOF}

\begin{LONG}
  \begin{minipage}[c]{.45\textwidth}
    \centering
    \includegraphics[scale=1]{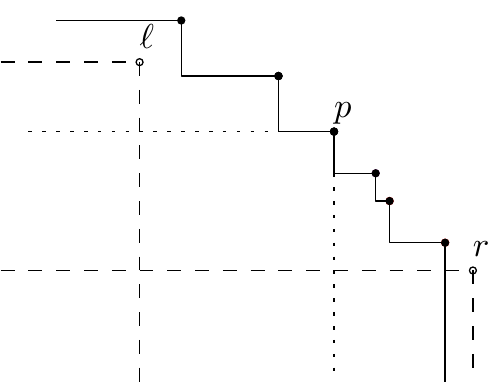}
  \end{minipage}\hfill
  \begin{minipage}[c]{.45\textwidth}
    \captionof{figure}{A representation of a state of the \texttt{Quick Union Algorithm} where the points $p$, $\ell$ and $r$ has been computed.}
    \label{fig:instance}
  \end{minipage}~\\
\end{LONG}

In the following section we describe a synergistic result that combines the results of Sections~\ref{sec:inputOrderAdaptivePlanarMaxima} and~\ref{sec:quick-maxima}. This result introduces the synergistic technique also used in the computation of the \textsc{Convex Hull} in Section~\ref{sec:synergisticUpperHull}.

\subsection{Synergistic Computation of Maxima Sets}
\label{sec:synergMaxima}

The \texttt{(Smooth,Structure) Synergistic Maxima} algorithm decomposes the input of planar points into the minimal number $\sigma$ of smooth subsequences of consecutive positions, computes their maxima sequences and then merges them using the \texttt{Quick Union Maxima} algorithm described in the previous section.

\begin{theorem}\label{theo:syn-max}
  Let $\mathcal{S}$ be a set of points in the plane such that ${\cal S}$ can be divided into $\sigma$ smooth maxima sequences. Let $h$ be the number of points in the \textsc{Maxima Set} of ${\cal S}$. There exists an algorithm that performs within $2n + O(\sum_{j=1}^{\beta}\log{s_{j}} + \sum^{\delta}_{i=1}\log{\binom{\sigma}{m_i}}) \subseteq O(n\log(\min(\sigma, h)))$~\footnote{The quantity $\sum_{j=1}^{\beta}\log{s_{j}}$ is within $O(n)$ but is much smaller for ``easy'' instances.} data comparisons when it computes the \textsc{Maxima Set} of $\mathcal{S}$; where $\beta$ and $s_1, \dots, s_\beta$ are the number and sizes of the blocks in the certificate ${\cal C}$ computed by the union algorithm, respectively; $\delta$ is the length of ${\cal C}$; and $m_1, \dots, m_\delta$ is a sequence where $m_i$ is the number of maxima sequences whose blocks form the $i$-th argument of ${\cal C}$. This number of comparisons is optimal in the worst case over instances $\mathcal{S}$ formed by $\sigma$ smooth sequences which \textsc{Maxima Set} have certificates ${\cal C}$ of length $\delta$ formed by $\beta$ blocks of sizes $s_1, \dots, s_\beta$, such that $m_1, \dots, m_\delta$ is a sequence where $m_i$ is the number of maxima sequences whose blocks form the $i$-th argument of ${\cal C}$.
\end{theorem}

\begin{VLONG}
  We prove that the number of comparisons performed by the algorithm \texttt{(Smooth, Structure) Synergistic Maxima} is asymptotically optimal in the worst case over instances formed by $n$ points grouped in $\sigma$ smooth sequences, with a final \textsc{Maxima Set} of size $h$.
  The upper bound is a consequence of the Theorem~\ref{theo:qum} and the linear time partitioning algorithm described in Section~\ref{sec:inputOrderAdaptivePlanarMaxima}.  We describe the intuition for the lower bound below: it is a simple adversary argument, based on the definition of a family of ``hard'' instances for each possible value of the parameters of the analysis, building over each other.

  First, we verify the lower bound for ``easy'' instances, of finite difficulty:
  general instances formed by a single ($\sigma=1$) smooth sequence obviously require $\Omega(n)$ comparisons (no correct algorithm can afford to ignore a single point of the input, which could dominate all others), while
  general instances dominated by a single point ($h=1$) also require $\Omega(n)$ comparisons (similarly to the computation of the maximum of an unsorted sequence).
  Each of this lower bound yields a distribution of instances, either of smoothness $\sigma=1$ or of output size $h=1$, such that any deterministic algorithm performs $\Omega(n)$ comparisons on average on a uniform distribution of those instances.

  Such distributions of ``elementary'' instances can be duplicated so that to produce various distributions of elementary instances; and combined so that to define a distribution of harder instances:
  \begin{lemma}
    Given the positive integers $n,\sigma, \beta, s_1, \ldots, s_\beta, \delta, m_1, \ldots, m_\delta$,
    there is a family of instances which can each be partitioned into $\sigma$ smooth sequences such that,
    \begin{itemize}
    \item $\beta$ and $s_1, \ldots, s_\beta$ are the number and sizes of the blocks in the certificate ${\cal C}$ computed by the union algorithm, respectively;
    \item $\delta$ is the length of ${\cal C}$;
    \item $m_1, \ldots, m_\delta$ is a sequence where $m_i$ is the number of maxima sequences of the smooth sequences whose blocks form the $i$-th argument of ${\cal C}$; and
    \item on average on a uniform distribution of these instances, any algorithm computing the \textsc{Maxima Set} of $S$ in the comparison model performs within $\Omega(n + \sum^{\delta}_{i=1}\log{\binom{\sigma}{m_i}})$ comparisons.
    \end{itemize}
  \end{lemma}

  Finally, any such distribution with a computational lower bound on average yields a computational lower bound for the worst case instance complexity of any randomized algorithm, on average on its randomness; and as a particular case a lower bound on the worst case complexity of any deterministic algorithm:

\begin{corollary}
  Given the positive integers $\sigma, \beta, s_1, \ldots, s_\beta, \delta, m_1, \ldots, m_\delta$, and an algorithm $A$ computing the \textsc{Maxima Set} of a sequence of $n$ planar points in the comparison model (whether deterministic or randomized), there is an instance $I$ such that $A$ performs a number of comparisons within $\Omega(n + \sum^{\delta}_{i=1}\log{\binom{\sigma}{m_i}})$ when it computes the \textsc{Maxima Set} of $I$.
\end{corollary}
\begin{proof}
  \begin{LONG}
    A direct application of Yao's minimax principle \cite{1958-PJM-OnGeneralMinimaxTheorems-Sion,1977-FOCS-ProbabilisticComputationsTowardAUnifiedMeasureOfComplexity-Yao,1944-BOOK-TheoryOfGamesAndEconomicBehavior-VonNeumannMorgenstern}.  \qed
  \end{LONG}
  \begin{SHORT}
    A direct application of Yao's minimax principle \cite{1958-PJM-OnGeneralMinimaxTheorems-Sion}. \qed
  \end{SHORT}
\end{proof}
\end{VLONG}

The histories of the computation of the \textsc{Maxima Set} and of the computation of the \textsc{Convex Hull} are strongly correlated: most of the results on one problem also generalize on the other one. Our results on the computation of the \textsc{Maxima Set} similarly generalize to the computation of the \textsc{Convex Hull}, albeit it requires additional work and concepts, which we describe in the next section.

\section{Convex Hull}
\label{sec:convex}

Given a set ${\cal S}$ of planar points, the \textsc{Convex Hull} of ${\cal S}$ is the smallest convex set containing ${\cal S}$~\cite{1985-BOOK-ComputationalGeometryAnIntroduction-PreparataShamos}.
Given $n$ points in the plane, the problem of computing their \textsc{Convex Hull} is well studied: the worst case complexity over instances of size $n$ is within $\Theta(n\lg n)$ in the algebraic decision tree computational model~\cite{1977-CACM-ConvexHullsOfFiniteSetsOfPointsInTwoAndThreeDimensions-PreparataHong}. Several refinements of this analysis are known: some taking advantage of the input structure~\cite{1985-SOCG-OutputSizeSensitiveAlgorithmsForFindingMaximalVectors-KirkpatrickSeidel, 2009-FOCS-InstanceOptimalGeometricAlgorithms-AfshaniBarbayChan} and some taking advantage of the input order~\cite{2011-IEICE-AdaptiveAlgorithmsForPlanarConvexHullProblems-AhnOkamoto,2002-SWAT-AdaptiveAlgorithmsForConstructingConvexHullsAndTriangulationsOfPolygonalChains-LevcopoulosLingasMitchell}.

\begin{INUTILE}
  Inspired by the results on the computation of \textsc{Maxima Sets} taking advantage of both the input order and the input structure presented in Section~\ref{sec:maxima}, we present similar results on the computation of the \textsc{Convex Hull} which take advantage both of the input order (as defined by Levcopoulos \emph{et al.}~\cite{2002-SWAT-AdaptiveAlgorithmsForConstructingConvexHullsAndTriangulationsOfPolygonalChains-LevcopoulosLingasMitchell}) and of the input structure (as defined by Afshani \emph{et al.}~\cite{2009-FOCS-InstanceOptimalGeometricAlgorithms-AfshaniBarbayChan}) at the same time in a synergistic way.
\end{INUTILE}

Levcopoulos et al.~\cite{2002-SWAT-AdaptiveAlgorithmsForConstructingConvexHullsAndTriangulationsOfPolygonalChains-LevcopoulosLingasMitchell} described how to \emph{partition} a sequence of points into subsequences of consecutive positions for which the \textsc{Convex Hull} can be computed in linear time. We refine their analysis to take into account the distribution of the sizes of the subsequences (Section~\ref{sec:inputOrderAdaptiveConvexHull}). This notion of input order for the computation of the \textsc{Convex Hull} is less restrictive than the one seen for the computation of the \textsc{Maxima Set}, in the sense that it allows to consider more sophisticated sequences as ``easy'' sequences.

As the computation of \textsc{Convex Hulls} reduces to the computation of \textsc{Upper Hulls}\begin{LONG} (the computation of the \textsc{Lower Hull} is symmetric and completes it into the computation of the \textsc{Convex Hull})\end{LONG}, we focus on the latter.
We describe an algorithm \textsc{Merging Upper Hulls} in Section~\ref{sec:UpperHullUnion}, which yields a synergistic algorithm taking advantage of both the input order and the input structure in Section~\ref{sec:synergisticUpperHull}. This synergistic algorithm outperforms both the algorithms described by Levcopoulos et al.~\cite{2002-SWAT-AdaptiveAlgorithmsForConstructingConvexHullsAndTriangulationsOfPolygonalChains-LevcopoulosLingasMitchell} and Afshani et al.~\cite{2009-FOCS-InstanceOptimalGeometricAlgorithms-AfshaniBarbayChan}, as well as any dovetailing combination of them.

\begin{TODO}
CHECKOUT \cite{1990-BIT-ASublogarithmicConvexHullAlgorithm-FjallstromKatajainenLevcopoulosPetersson}, I don't remember what results it has?
\end{TODO}

\subsection{Input Order Adaptive Convex Hull}
\label{sec:inputOrderAdaptiveConvexHull}

A \emph{polygonal chain} is a curve specified by a sequence of points $p_1, \dots, p_n$. The curve itself consists of the line segments connecting the pairs of consecutive points. A polygonal chain is \emph{simple} if
\begin{SHORT}
  it does not have a self-intersection.
\end{SHORT}
\begin{LONG}
  any two edges of $P$ that are not adjacent are disjoint, or if the intersection point is a vertex of $P$; and any two adjacent edges share only their common vertex.
\end{LONG}
Levcopoulos et al.~\cite{2002-SWAT-AdaptiveAlgorithmsForConstructingConvexHullsAndTriangulationsOfPolygonalChains-LevcopoulosLingasMitchell} described an algorithm that computes the \textsc{Convex Hull} of $n$ planar points in time within $O(n\log\kappa)$, where $\kappa$ is the minimal number of simple chains into which the input sequence of points can be partitioned. The algorithm partitions the points into simple subchains, computes their \textsc{Convex Hulls}, and merges them. In their analysis the complexity of both the partitioning and merging steps are within $O(n\log\kappa)$.
In Section~\ref{sec:synergisticUpperHull}, we describe a partitioning algorithm running in linear time, which is key to the synergistic result.
\begin{LONG}
  
  For a given polygonal chain, there can be several partitions into simple subchains of minimum size for it.
\end{LONG}
We describe a refined analysis which takes into account the relative imbalance between the sizes of the subchains.
\begin{LONG}
  The idea behind the refinement is to bound the number of operations that the algorithm executes for every simple subchain.
  This analysis makes it possible to identify families of instances where the complexity of the algorithm is linear even though the number of simple subchains into which the chain is split is logarithmic.
\end{LONG}
\begin{LONG}
  In the recursion tree of the execution of the algorithm described by Levcopoulos \emph{et al.}~\cite{2002-SWAT-AdaptiveAlgorithmsForConstructingConvexHullsAndTriangulationsOfPolygonalChains-LevcopoulosLingasMitchell} on an input $C$ formed by $n$ planar points, every node represents a subchain of $C$. The cost of every node is linear in the size of the subchain that it represents. The simplicity test and the merge process are both linear in the number of points in the subchain. Every time this algorithm discovers that the polygonal chain is simple, the corresponding node in the recursion tree becomes a leaf.
\end{LONG}

\begin{theorem}
\label{theo:simple}
Given a sequence $S$ of $n$ planar points which can be partitioned into $\kappa$ simple subchains of respective sizes $n_1,\ldots,n_\kappa$, Levcopoulos et al.'s algorithm~\cite{2002-SWAT-AdaptiveAlgorithmsForConstructingConvexHullsAndTriangulationsOfPolygonalChains-LevcopoulosLingasMitchell}  computes the convex hull of $S$ in time  within $O(n(1+\mathcal{H}(n_1, \dots, n_{\kappa}))) \subseteq O(n(1{+}\log{\kappa})) \subseteq O(n\log{n})$, where $\mathcal{H}(n_1, \dots, n_\kappa) = \sum_{i=1}^\kappa{\frac{n_i}{n}}\log{\frac{n}{n_i}}$\begin{SHORT}.\end{SHORT}
\begin{LONG}
, which is worst-case optimal over instances of $n$ points that can be partitioned into $\kappa$ simple subchains of sizes $n_1, \dots, n_{\kappa}$.
\end{LONG}
\end{theorem}
\begin{INUTILE}
  \begin{theorem}
\label{theo:simple}
Given a sequence $S$ of $n$ planar points, Levcopoulos et al.'s algorithm~\cite{2002-SWAT-AdaptiveAlgorithmsForConstructingConvexHullsAndTriangulationsOfPolygonalChains-LevcopoulosLingasMitchell} computes the convex hull of $S$ in time within $O(n(1+\alpha)) \subseteq O(n(1{+}\log{\kappa})) \subseteq O(n\log{n})$, where $\alpha$ is the minimal entropy $\min\{\mathcal{H}(n_1, \dots, n_{\kappa})$ any partition of $S$ into
 $\kappa$ simple subchains of consecutive positions, of respective sizes $n_1,\ldots,n_\kappa}\}$  and $\mathcal{H}(n_1, \dots, n_\kappa) = \sum_{i=1}^\kappa{\frac{n_i}{n}}\log{\frac{n}{n_i}}$.
\begin{LONG}
  , which is worst-case optimal over instances of $n$ points that can be partitioned into $\kappa$ simple subchains of sizes $n_1, \dots, n_{\kappa}$.
\end{LONG}
\end{theorem}
\end{INUTILE}
\begin{PROOF}
  \begin{proof}
    Fix the subchain $c_i$ of size $n_i$. In the worst case, the algorithm considers the $n_i$ points of $c_i$ for the simplicity test and the merging process, in all the levels of the recursion tree from the first level to the level $\lceil \log{\frac{n}{n_i}} \rceil + 1$, because the sizes of the subchains in these levels are greater than $n_i$. In the next level, one of the nodes $\ell$ of the recursion tree fits completely inside $c_i$ and therefore it becomes a leaf. Hence, at least $\frac{n_i}{4}$ points from $c_i$ are discarded for the following iterations. The remaining points of $c_i$ are in the left or the right ends of subchains represented by nodes in the same level of $\ell$ in the recursion tree. In all of the following levels, the number of operations of the algorithm involving points from $c_i$ can be bounded by the size of the subchains in those levels. So, the sum of the number of the operations in these levels is within $O(n_i)$.  As a result, the number of operations of the algorithm involving points from $c_i$ is within $O(n_i\log{\frac{n}{n_i}} + n_i)$. In total, the time complexity of the algorithm is within $O(n + \sum_{i=1}^\kappa n_i\log{\frac{n}{n_i}}) = O(n(1+\mathcal{H}(n_1, \dots, n_\kappa))) \subseteq O(n(1{+}\log{\kappa})) \subseteq O(n\log{n})$.

We prove the optimality of this complexity in the worst-case over instances of $n$ points that can be partitioned into $\kappa$ simple subchains of sizes $n_1, \dots, n_{\kappa}$ by giving a tight lower bound. Barbay and Navarro~\cite{2013-TCS-OnCompressingPermutationsAndAdaptiveSorting-BarbayNavarro} showed a lower bound of $\Omega(n(1+{\cal H}(r_1,\ldots,r_\rho)))$ in the comparison model for {\sc{Sorting}} a sequence of $n$ numbers, in the worst case over instances covered by $\rho$ runs (increasing or decreasing) of sizes $r_1, \dots, r_\rho$, respectively, summing to $n$.  The {\sc{Sorting}} problem can be reduced in linear time to the problem of computing the {\sc{Convex Hulls}} of a chain of $n$ planar points that can be partitioned into $\rho$ simple subchains of sizes $r_1, \dots, r_\rho$, respectively. For each real number $r$, this is done by producing a point with $(x,y)$-coordinates $(r,r^2)$. The $\rho$ runs (alternating increasing and decreasing) are transformed into $\rho$ simple subchains of the same sizes. The sorted sequence of the numbers can be obtained from the {\sc{Convex Hull}} of the points in linear time. \qed
  \end{proof}
\end{PROOF}

Similarly to the computation of the \textsc{Maxima Set}, we define in the following section an algorithm for \textsc{Merging Upper Hulls}. This algorithm is a building block towards the synergistic algorithm that computes the \textsc{Convex Hull} of a set of planar points, and is more complex than that for \textsc{Merging Maxima}.

\subsection{Union of Upper Hulls}
\label{sec:UpperHullUnion}

We describe the \texttt{Quick Union Hull} algorithms which computes the \textsc{Upper Hull} of the union of $\rho$ upper hull sequences in the plane assuming that the upper hull sequences are given in sorted order by their $x$-coordinates. Given an upper hull sequence $\mathcal{U}_i$, let $\mathcal{U}_i[a]$ and $\mathcal{U}_i[b..c]$ denote the $a$-th point and the block of $c-b+1$ consecutive points corresponding to the positions from $a$ to $b$ in $\mathcal{U}_i$, respectively.  Given two points $p$ and $q$, let $m(p,q)$ denote the slope of the straight line that passes trough $p$ and $q$.

\subsubsection{Description of the algorithm Quick Union Hull.}

The \texttt{Quick Union Hull} algorithm is inspired by an algorithm described by Chan et al.~\cite{1997-DCG-PrimalDividingAndDualPruningOutputSensitiveConstructionOfFoudDimensionalPolytopesAndThreeDimensionalVoronoiDiagrams-ChanSnoeyinkYap}. It chooses an edge of slope $\mu$ from the upper hull sequences, and computes the point $p$ that has a supporting line of slope $\mu$. The algorithm then splits the points in the upper hull sequences by $p_x$. It computes the two tangents of $p$ with all the upper hull sequences: the one to the left of $p$ and the one to the right of $p$, and discards all the points below these tangents. The algorithm also computes a block of consecutive points in the upper hull sequence that contains $p$ which points are part of the output, and discards the points in this block ($p$ is in this block). (This last step is key to the optimality of the algorithm and is significantly more complex than its counterpart in the \textsc{Merging Maxima} solution.) The algorithm then recurses on the non-discarded points to the left of $p$ and on the non-discarded points to the right of $p$. All these steps take advantage of the fact that the points in the upper hull sequences are sorted in increasing order of their $x$-coordinates and the slopes of the edges of the upper hull sequences are monotonically decreasing from left to right. (See Algorithm~\ref{alg:quh} for a more formal description.)

\begin{algorithm}[t] 
  \caption{\texttt{Quick Union Hull}} 
  \label{alg:quh} 
  \begin{algorithmic}[1] 

    \REQUIRE{A set $\mathcal{U}_1, \dots, {\cal U}_\rho$ of $\rho$ upper hull sequences}
    \ENSURE{The \textsc{Upper Hull} of the union of the set $\mathcal{U}_1, \dots, {\cal U}_\rho$}
    
    \STATE Compute the median $\mu$ of the slopes of the middle edges of the
    upper hull sequences;

    \STATE Find the point $p$ that has a supporting line of slope $\mu$ through
    doubling searches for the value $\mu$ in the slopes of the edges of all upper hull sequences,
    starting at both ends in parallel, note $j\in[1..\rho]$
    the index of the upper hull sequence containing $p$;
    
    \STATE Perform doubling searches for the value $p_x$ in the
    $x$-coordinates of the points of all upper hull sequences except ${\cal U}_j$, starting at both
    ends in parallel;
    
    \STATE Find the two tangents of $p$ with all upper hull sequences: the one to the left of $p$
    and the one to the right of $p$, through doubling searches testing for each point $q$ the slope
    of the two edge that have $q$ as end point and the slope of the line $pq$ and discard
    the points below these tangents.

    \STATE Discard a block in ${\cal U}_j$ containing $p$ that form part of the output by
    computing the tangent between ${\cal U}_j$ and the upper hull sequences left of $p$ of
    minimum slope and the tangent between ${\cal U}_j$ and the upper hull sequences right of $p$
    of maximum slope.

    \STATE Repeat until there is no more than one upper hull sequence of size $1$:
    pair those left of $p$ and pair those right of $p$ and apply the
    Step $4$ to those pairs;

    \STATE Discard all points that lie below the lines that joins $p$ with the leftmost
    point and the rightmost point of the upper hull sequences;
    
    \STATE Recurse on the non-discarded points left and right of $p$.
  \end{algorithmic}
\end{algorithm}

In the following we describe in more details the Step $5$ of Algorithm~\ref{alg:quh}. We describe only how to compute a block of consecutive points to the right of $p$ in ${\cal U}_j$ that form part of the output, as the left counterpart is symmetric. Let $\tau$ be the tangent of maximum slope between ${\cal U}_j$ and the upper hull sequences to the right of $p$. Let $q$ be the point in ${\cal U}_j$ that lies in $\tau$.  Let $\lambda$ be the tangent of maximum slope among those computed at Step $4$.
\begin{LONG}
  All the points to the right of $p$ are below $\lambda$ and
\end{LONG}
$\lambda$ is a separating line between the portion of ${\cal U}_j$ that contains $q$ and the points to the right of $p$. Given two upper hull sequences ${\cal U}_i$ and ${\cal U}_k$ separated by a vertical line,  Barbay and Chen~\cite{2008-CCCG-ConvexHullOfTheUnionOfConvexObjectsInThePlane-BarbayChen} described an algorithm that computes the common tangent between ${\cal U}_i$ and ${\cal U}_k$ in time within $O(\log a + \log b)$, where ${\cal U}_i[a]$ and ${\cal U}_k[b]$ are the points that lie in the tangent. At each step this algorithm considers the points ${\cal U}_i[c]$ and ${\cal U}_k[d]$ and can certify at least in one upper hull sequence if the tangent is to the right or to the left of the point considered. A minor variant manages the case where the separating line is not vertical. Algorithm~\ref{alg:quh} executes several instances of this algorithm in parallel between ${\cal U}_j$ and all upper hull sequences to the right of $p$, always considering the same point in ${\cal U}_j$ (similarly to the Demaine et al.'s algorithm~\cite{2000-SODA-AdaptiveSetIntersectionsUnionsAndDifferences-DemaineLopezOrtizMunro} to compute the intersection of sorted set). Once all parallel decisions about the point ${\cal U}_j[c]$ are made, the instances can be divided into two sets: (i) those whose tangents are to the left of ${\cal U}_j[c]$ and (ii) those whose tangents are to the right of ${\cal U}_j[c]$. Algorithm~\ref{alg:quh} stops the parallel computation of tangents in the set of maxima sequences (ii). The Step $5$ continues until there is just one instance running and computes the tangent $\tau$ in this instance.

\subsubsection{Analysis of the Quick Union Hull Algorithm.}
\label{sec:analysisQUH}


Similarly to the case of \textsc{Merging Maxima}, every algorithm for \textsc{Merging Upper Hulls} needs to certify that some blocks of the upper hull sequences can not participate in the \textsc{Upper Hull} of the union, and that some blocks of the upper hull sequences are in the \textsc{Upper Hull} of the union.
\begin{INUTILE}
  that computes the \textsc{Upper Hull} of the union of $\rho$ upper hull sequences needs to certify that some blocks of those can not participate in the \textsc{Upper Hull} of the union, and that some blocks of those are in the \textsc{Upper Hull} of the union.
\end{INUTILE}
In the following we formalize the notion of a \emph{certificate} for \textsc{Merging Upper Hulls} problem.

\begin{definition}
  Given the points $\mathcal{U}_i[a]$ and $\mathcal{U}_j[b]$, let $\ell$ be the straight line that passes through $\mathcal{U}_i[a]$ and $\mathcal{U}_j[b]$ and let $m_\ell$ be the slope of $\ell$. $\langle \mathcal{U}_i[a], \mathcal{U}_j[b] \supset \mathcal{U}_k[c..d..e] \rangle$ is an \emph{Elementary Eliminator Argument} if all the points of the block $\mathcal{U}_k[c..e]$ are between the vertical lines through $\mathcal{U}_i[a]$ and $\mathcal{U}_j[b]$, $m(\mathcal{U}_k[d-1], \mathcal{U}_k[d]) \ge m_\ell \ge m(\mathcal{U}_k[d], \mathcal{U}_k[d+1])$, and the point $\mathcal{U}_k[d]$ lies below $\ell$.
\end{definition}

If $\langle \mathcal{U}_i[a], \mathcal{U}_j[b] \supset \mathcal{U}_k[c..d..e] \rangle$ is an elementary eliminator argument then the points in the block $\mathcal{U}_k[c..e]$ can not participate in the \textsc{Upper Hull} of the union.
\begin{LONG}
  \begin{lemma}
    An elementary eliminator argument $\langle \mathcal{U}_i[a], \mathcal{U}_j[b] \supset \mathcal{U}_k[c..d..e] \rangle$ can be checked in constant time.
  \end{lemma}
\end{LONG}
Several blocks that are ``eliminated'' by the same pair of points can be combined into a single argument, a notion captured by the \emph{block eliminator argument}.

\begin{definition}
  $\langle \mathcal{U}_i[a], \mathcal{U}_j[b] \supset \mathcal{U}_{k_1}[c_1..d_1..e_1], \dots, \mathcal{U}_{k_t}[c_t..d_t..e_t] \rangle$ is a \emph{Block Eliminator Argument} if $\langle \mathcal{U}_i[a], \mathcal{U}_j[b] \supset \mathcal{U}_{k_1}[c_1..d_1..e_1] \rangle, \dots, \langle \mathcal{U}_i[a], \mathcal{U}_j[b] \supset \mathcal{U}_{k_t}[c_t..d_t..e_t] \rangle$ are elementary eliminator arguments.
\end{definition}

\begin{LONG}
A block eliminator argument is checked by checking each elementary eliminator argument that form it.
  
  \begin{corollary}
    A block eliminator argument $\langle \mathcal{U}_i[a], \mathcal{U}_j[b] \supset \mathcal{U}_{k_1}[c_1..d_1..e_1], \dots,\mathcal{U}_{k_t}[c_t..d_t..e_t] \rangle$ can be checked in time within $O(t)$.
  \end{corollary}
\end{LONG}

As for \textsc{Merging Maxima}, any correct algorithm for \textsc{Merging Upper Hulls} must certify that some points are part of the output.

\begin{definition}
  $\langle \mathcal{U}_i[a] \dashv \mathcal{U}_{j_1}[b_1], \dots, \mathcal{U}_{j_t}[b_t] \rangle$ is an \emph{Elementary Convex Argument} if there exists a straight line $\ell$ that passes through $\mathcal{U}_i[a]$ of slope $m_\ell$ such that $m(\mathcal{U}_{j_1}[b_1-1], \mathcal{U}_{j_1}[b_1]) \ge m_\ell \ge m(\mathcal{U}_{j_1}[b_1], \mathcal{U}_{j_1}[b_1+1]), \dots, m(\mathcal{U}_{j_t}[b_t-1], \mathcal{U}_{j_t}[b_t]) \ge m_\ell \ge m(\mathcal{U}_{j_t}[b_t], \mathcal{U}_{j_t}[b_t+1])$; $m(\mathcal{U}_i[a-1], \mathcal{U}_i[a]) \ge m_\ell \ge m(\mathcal{U}_i[a], \mathcal{U}_i[a+1])$; and the points $\mathcal{U}_{j_1}[b_1], \dots, \mathcal{U}_{j_t}[b_t]$ lie below $\ell$.
\end{definition}

If $\langle \mathcal{U}_i[a] \dashv \mathcal{U}_{j_1}[b_1], \dots, \mathcal{U}_{j_t}[b_t] \rangle$ is an elementary convex argument, then the point $\mathcal{U}_i[a]$ is in the \textsc{Upper Hull} of the union of the upper hulls $\mathcal{U}_i, \mathcal{U}_{j_1}, \dots, \mathcal{U}_{j_t}$.
\begin{LONG}
  \begin{lemma}
    An \emph{elementary convex argument} $\langle \mathcal{U}_i[a] \dashv \mathcal{U}_{j_1}[b_1], \dots, \mathcal{U}_{j_t}[b_t] \rangle$ can be checked in time within $O(t)$.
  \end{lemma}
\end{LONG}
There are blocks that can be ``easily'' certified that form part of the output.

\begin{definition}\label{def:blockConvex}
  Given the points $\mathcal{U}_i[a]$ and $\mathcal{U}_i[b]$, let $\ell$ be the straight line that passes through $\mathcal{U}_i[a]$ and $\mathcal{U}_i[b]$ and let $m_\ell$ be the slope of $\ell$.  $\langle \mathcal{U}_i[a..b] \dashv \mathcal{U}_{j_1}[c_1..d_1..e_1],\begin{SHORT}\\\end{SHORT} \dots, \mathcal{U}_{j_t}[c_t..d_t..e_t] \rangle$ is a \emph{Block Convex Argument} if $\langle \mathcal{U}_i[a] \dashv \mathcal{U}_{j_1}[c_1], \dots, \mathcal{U}_{j_t}[c_t] \rangle$ and $\langle \mathcal{U}_i[b] \dashv \mathcal{U}_{j_1}[e_1], \dots, \mathcal{U}_{j_t}[e_t] \rangle$ are elementary convex arguments; $m(\mathcal{U}_{j_1}[d_1-1], \mathcal{U}_{j_1}[d_1]) \ge m_\ell \ge m(\mathcal{U}_{j_1}[d_1], \mathcal{U}_{j_1}[d_1+1]),\dots, m(\mathcal{U}_{j_t}[d_t-1], \mathcal{U}_{j_t}[d_t]) \ge m_\ell \ge m(\mathcal{U}_{j_t}[d_t], \mathcal{U}_{j_1}[d_t+1])$, and the points $\mathcal{U}_{j_1}[d_1], \dots, \mathcal{U}_{j_t}[d_t]$ lie below $\ell$.
\end{definition}

If $\langle \mathcal{U}_i[a..b] \dashv \mathcal{U}_{j_1}[c_1..d_1..e_1], \dots, \mathcal{U}_{j_t}[c_t..d_t..e_t] \rangle$ is a block convex argument then the points in the block $\mathcal{U}_i[a..b]$ are in the \textsc{Upper Hull} of the union of the upper hulls $\mathcal{U}_i, \mathcal{U}_{j_1}, \dots, \mathcal{U}_{j_t}$.

\begin{LONG}
  \begin{lemma}
    A \emph{block convex argument} $\langle \mathcal{U}_i[a..b] \dashv \mathcal{U}_{j_1}[c_1..d_1..e_1], \dots, \mathcal{U}_{j_t}[c_t..d_t..e_t] \rangle$ can be checked in time within $O(t)$.
  \end{lemma}

   Similar to the \textsc{Merging Maxima}, the difficulty of finding and describing block eliminator and block convex arguments depend on the points they refer to in the upper hull sequences, a notion captured by ``argument points'':
\end{LONG}

\begin{definition}
  Given an argument $\langle \mathcal{U}_i[a], \mathcal{U}_j[b] \supset \mathcal{U}_{k_1}[c_1..d_1..e_1], \dots, \mathcal{U}_{k_t}[c_t..d_t..e_t] \rangle$ or $\langle \mathcal{U}_i[a..b] \dashv \mathcal{U}_{j_1}[c_1..d_1..e_1], \dots, \mathcal{U}_{j_t}[c_t..d_t..e_t] \rangle$, the \emph{Argument Points} are the points $\mathcal{U}_i[a]$ and $\mathcal{U}_i[b]$.
\end{definition}

Those atomic arguments can be checked in time proportional to the number of blocks in them, and combine into a general definition of a certificate that any correct algorithm for \textsc{Merging Upper Hulls} in the algebraic decision tree computational model can be modified to output.

\begin{definition}
  Given a set of upper hull sequences and their \textsc{Upper Hull} $\mathcal{U}$ expressed as several blocks on the upper hull sequences. A \emph{certificate} of $\mathcal{U}$ is a set of block eliminator and block convex arguments such that the \textsc{Upper Hull} of any instance satisfying those arguments is given by the description of $\mathcal{U}$. The length of a certificate is the number of distinct argument points in it.
\end{definition}

Similarly to the \textsc{Merging Maxima}, the key of the analysis is to separate the doubling search steps from the other steps of the algorithm.   

\begin{theorem}\label{theo:quh}
  Given $\rho$ upper hull sequences. The time complexity of the \texttt{Quick Union Hull} algorithm is within $O(\sum_{j=1}^{\beta}\log{s_{j}} + \sum^{\delta}_{i=1}\log{\binom{\rho}{m_i}})$ when it computes the \textsc{Upper Hull} of the union of these upper hull sequences; where $\beta$ is the number of blocks in the certificate ${\cal C}$ computed by the algorithm; $s_1, \dots, s_\beta$ are the sizes of these blocks; $\delta$ is the length of ${\cal C}$; and $m_1, \dots, m_\delta$ is a sequence where $m_i$ is the number of upper hull sequences whose blocks form the $i$-th argument of ${\cal C}$.
\end{theorem}

The optimality of this algorithm is a consequence of the fact that it checks each argument of any certificate using a constant number of argument points.

\begin{lemma}\label{lem:opt-hull}
  The algorithm \texttt{Quick Union Hull} computes a certificate
  which length is a constant factor of the length of the certificate of minimal length.
\end{lemma}

In the following section we describe a synergistic results that combines the results of Sections~\ref{sec:inputOrderAdaptiveConvexHull} and~\ref{sec:UpperHullUnion}.

\subsection{Synergistic Computation of Upper Hulls}
\label{sec:synergisticUpperHull}

The \texttt{(Simple,Structure) Synergistic Hull} algorithm proceeds in two phases. It first decomposes the input into simple subchains of consecutive positions using a (new) linear time doubling search inspired partitioning algorithm, that searches for simple chains from left to right (see \begin{LONG}Algorithm~\ref{alg:dsp}\end{LONG}\begin{SHORT}Appendix~\ref{app:simple}\end{SHORT} for a detail description of the algorithm). It computes their upper hull sequences, and then merges those using the \texttt{Quick Union Hull} algorithm described previously.

\begin{LONG}
  \begin{algorithm} 
    \caption{\texttt{Doubling Search Partition}} 
    \label{alg:dsp} 
    \begin{algorithmic}[1] 
      \REQUIRE{A sequence of $n$ planar points $p_1, \dots, p_n$} \ENSURE{A sequence of simple polygonal chains}
    
      \STATE Initialize $i$ to $1$;
    
      \FOR{$t = 1, 2, \dots$ } \IF{$i+2^t > n$ \OR the chain $p_i, \dots, p_{i+2^t}$ is \NOT simple} \STATE {Add the chain $p_i, \dots, p_{i+2^{t-1}}$ to the output} \STATE {Update $i \leftarrow i+2^{t-1} + 1$} \STATE {Reset $t \leftarrow 1$}
      \ENDIF
      \ENDFOR
    \end{algorithmic}
  \end{algorithm}

  The \texttt{Doubling Search Partition} algorithm partitions the polygonal chain into simple subchains which sizes has asymptotically minimum entropy among all the partitions into simple subchains. The following lemma formalizes this fact.

\begin{lemma}
  Given a sequence $S$ of $n$ planar points. The \texttt{Doubling Search Partition} algorithm computes in linear time a partition of $S$ into $k$ simple polygonal chain of consecutive points, of sizes $n_1, \dots, n_k$ such that $n(1+\mathcal{H}(n_1, \dots, n_{k})) \in O(n(1+\alpha))$, where $\alpha$ is the minimal entropy $\min\{\mathcal{H}(n_1, \dots, n_{\kappa})$ of any partition of $S$ into $\kappa$ simple subchains of consecutive positions, of respective sizes $n_1,\ldots,n_\kappa$, and $\mathcal{H}(n_1, \dots, n_\kappa) = \sum_{i=1}^\kappa{\frac{n_i}{n}}\log{\frac{n}{n_i}}\}$.
\end{lemma}

The proof of this lemma is similar to the proof of Theorem~\ref{theo:simple} where the number of operations for each simple subchain of a partition into simple subchains is bounded separately. The following theorem summarizes the synergistic result in this section.

\end{LONG}

\begin{theorem}\label{theo:syn-hull}
  Let $\mathcal{S}$ be a set of points in the plane such that ${\cal S}$ can be partitioned into $\kappa$ simple subchains. Let $h$ be the number of points in the \textsc{Upper Hull} of ${\cal S}$. There exists an algorithm which time complexity is within $O(n + \sum^{\delta}_{i=1}\log{\binom{\kappa}{m_i}}) \subseteq O(n\log(\min(\kappa, h)))$ when it computes the \textsc{Upper Hull} of $\mathcal{S}$, where $\delta$ is the length of the certificate ${\cal C}$ computed by the union algorithm; and $m_1, \dots, m_\delta$ is a sequence where $m_i$ is the number of upper hull sequences of the simple subchains whose blocks form the $i$-th argument of ${\cal C}$. This number of comparisons is optimal in the worst case over instances $\mathcal{S}$ formed by $\kappa$ simple subchains which \textsc{Upper Hulls} have certificates ${\cal C}$ of length $\delta$ such that $m_1, \dots, m_\delta$ is a sequence where $m_i$ is the number of upper hull of simple subchains whose blocks form the $i$-th argument of ${\cal C}$.
\end{theorem}

\begin{VLONG}
  Even though the algorithms are more complex, except for some details, the proofs of Theorems~\ref{theo:quh} and~\ref{theo:syn-hull} and Lemma~\ref{lem:opt-hull} of Section~\ref{sec:convex} are very similar to those described in the previous section.
  We describe the intuition for the lower bound below: as for the computation of \textsc{Maxima Sets}, it is a simple adversary argument, based on the definition of a family of ``hard'' instances for each possible value of the parameters of the analysis, building over each other, but the combination of elementary instances requires a little bit of extra care.

  First, we verify the lower bound for ``easy'' instances, of finite difficulty:
  general instances formed by a single ($\kappa=1$) simple sequence obviously require $\Omega(n)$ comparisons (no correct algorithm can afford to ignore a single point of the input), while
  general instances dominated by a single edge ($h=1)$ also require $\Omega(n)$ comparisons.
  Each of these lower bound yields a distribution of instances, either decomposed into $\kappa=1$ simple chains or of output size $h=1$, such that any deterministic algorithm performs $\Omega(n)$ comparisons on average on a uniform distribution of those instances.

  Such distributions of ``elementary'' instances can be duplicated so that to produce various distributions of elementary instances; and can be combined so that to define a distribution of harder instances.

\begin{lemma}
  Given the positive integers $n,\kappa, \beta, s_1, \ldots, s_\beta, \delta, m_1, \ldots, m_\delta$,
  there is a family of instances which can each be partitioned into $\kappa$ simple subchains such that,
  \begin{itemize}
  \item $\beta$ and $s_1, \ldots, s_\beta$ are the number and sizes of the blocks in the certificate ${\cal C}$ computed by the union algorithm, respectively;
  \item $\delta$ is the length of ${\cal C}$;
  \item $m_1, \ldots, m_\delta$ is a sequence where $m_i$ is the number of upper hulls of the simple subchains whose blocks form the $i$-th argument of ${\cal C}$; and
  \item on average on a uniform distribution of these instances, any algorithm computing the \textsc{Upper Hull} of $S$ in the comparison model performs within $\Omega(n + \sum^{\delta}_{i=1}\log{\binom{\kappa}{m_i}})$ comparisons.
  \end{itemize}
\end{lemma}

Finally, any such distribution with a computational lower bound on average yields a computational lower bound for the worst case instance complexity of any randomized algorithm, on average on its randomness; and as a particular case a lower bound on the worst case complexity of any deterministic algorithm:

\begin{corollary}
  Given the positive integers $\kappa, \beta, s_1, \ldots, s_\beta, \delta, m_1, \ldots, m_\delta$, and an algorithm $A$ computing the \textsc{Upper Hull} of a sequence of $n$ planar points in the algebraic decision tree computational model (whether deterministic or randomized), there is an instance $I$ such that $A$ performs a number of comparisons within $\Omega(n + \sum^{\delta}_{i=1}\log{\binom{\kappa}{m_i}})$ when it computes the \textsc{Upper Hull} of $I$.
\end{corollary}
\begin{proof}
  \begin{LONG}
    A direct application of Yao's minimax principle \cite{1958-PJM-OnGeneralMinimaxTheorems-Sion,1977-FOCS-ProbabilisticComputationsTowardAUnifiedMeasureOfComplexity-Yao,1944-BOOK-TheoryOfGamesAndEconomicBehavior-VonNeumannMorgenstern}.  \qed
  \end{LONG}
  \begin{SHORT}
    A direct application of Yao's minimax principle \cite{1958-PJM-OnGeneralMinimaxTheorems-Sion}.  \qed
  \end{SHORT}
\end{proof}
\end{VLONG}

\begin{LONG}
  This concludes the description of our synergistic results. In the next section, we discuss the issues left open for improvement.
\end{LONG}

\begin{LONG}
  \section{Discussion}
  \label{sec:discussion}

  Considering the computation of the \textsc{Maxima Set} and of the \textsc{Convex Hull}, we have built upon previous results taking advantage either of some notions of input order or of some notions of input structure, to describe solutions which take advantage of both in a synergistic way. There are many ways in which those results can be improved further: we list only a selection here.

\begin{INUTILE}
  First, Afshani et al.~\cite{2009-FOCS-InstanceOptimalGeometricAlgorithms-AfshaniBarbayChan} refined Kirkpatrick and Seidel's input structure adaptive results~\cite{1985-SOCG-OutputSizeSensitiveAlgorithmsForFindingMaximalVectors-KirkpatrickSeidel} for both the computation of the \textsc{Maxima Set} and of the \textsc{Convex Hull}: even those their solution is not of practical use (because of high constant factors), it would be interesting to obtain a synergistic solution which outperforms theirs.
\end{INUTILE}
In the same line of thought, Ahn and Okamoto~\cite{2011-IEICE-AdaptiveAlgorithmsForPlanarConvexHullProblems-AhnOkamoto} described some other notion of input order than the one we considered here, which can potentially yield another synergistic solution in combination with a given notion of input structure. This is true for any of the many notions of input order which could be adapted from \textsc{Sorting}~\cite{1992-ACJ-AnOverviewOfAdaptiveSorting-MoffatPetersson}.

Whereas being adaptive to as many measures of difficulty as possible at once is a worthy goal in theory, it usually comes at a price of an increase in the constant factor of the running time of the algorithm: it will become important to measure, for the various practical applications of each problem, which measures of difficulty take low value in practice. It will be necessary to do some more theoretical work to identify what to look for in the practical applications, but then it will be important to measure the practical difficulties of the instances.

\medskip \textbf{Acknowledgments:}
The authors would like to thank Javiel Rojas for helping with the bibliography on the computation of the \textsc{Maxima Set} of a set of points.
\end{LONG}


\bibliographystyle{splncs}
\bibliography{addedForThePaper,bibliographyDatabaseJeremyBarbay,publicationsJeremyBarbay}

\end{document}
